\documentclass[11pt]{article}
\usepackage[dvipdfmx]{graphicx}
\usepackage{here}
\usepackage{authblk}
\usepackage{amsmath,amssymb}
\usepackage{amsthm}
\usepackage{algorithm}
\usepackage{algorithmic}
\usepackage{dsfont}
\usepackage{booktabs}
\usepackage{color}

\newtheorem{dfn}{Definition} 
\newtheorem{lemma}{Lemma} 
\newtheorem{thm}{Theorem}

\newcommand{\tab}{\hspace*{3mm}}
\newcommand{\If}{{\bf if }}
\newcommand{\Then}{{\bf then }}
\newcommand{\Else}{{\bf else }}

\newcommand{\End}{{\bf end }}

\newcommand{\For}{{\bf for}}

\newcommand{\Goto}{{\bf goto}}
\newcommand{\Wait}{{\bf Wait}}

\newcommand{\Int}{{\bf int }}
\newcommand{\Uint}{{\bf uint }}

\usepackage{url}

\begin{document}
\title{
Atomic Cross-Chain Swaps with Improved Space and Local Time Complexity
}
%
%
\author[1]{Soichiro Imoto\thanks{Contact author. email:s-imoto@ist.osaka-u.ac.jp} }
\author[1]{Yuichi Sudo}
\author[2]{Hirotsugu Kakugawa}
\author[1]{Toshimitsu Masuzawa}

\affil[1]{Osaka University, Japan}
\affil[2]{Ryukoku University, Japan}
\date{}
\maketitle              

\begin{abstract}
An effective atomic cross-chain swap protocol is introduced by Herlihy [Herlihy, 2018] as a distributed coordination protocol in order to exchange assets across multiple blockchains among multiple parties.
An atomic cross-chain swap protocol guarantees;
(1) if all parties conform to the protocol, then all assets are exchanged among parties,
(2)even if some parties or coalitions of parties deviate from the protocol, no party conforming to the protocol suffers a loss, and
(3) no coalition has an incentive to deviate from the protocol.
Herlihy [Herlihy, 2018] invented this protocol by using hashed timelock contracts. \par

A cross-chain swap is modeled as a directed graph $D=(V,A)$. Vertex set $V$ denotes a set of parties and arc set $A$ denotes a set of proposed asset transfers.
Herlihy's protocol uses the graph topology and signature information to set appropriate hashed timelock contracts.
The space complexity of the protocol (i.e., the total number of bits written in the blockchains in a swap) is $O(|A|^2)$.
The local time complexity of the protocol (i.e., the maximum execution time of a contract in a swap to transfer the corresponding asset) is $O(|V|\cdot|L|)$, where $L$ is a feedback vertex set computed by the protocol.

We propose a new atomic cross-chain swap protocol which uses only signature information and improves the space complexity to $O(|A|\cdot|V|)$ and the local time complexity to $O(|V|)$.

\end{abstract}

\section{Introduction}
\subsection{Motivation}
The seminal work \cite{2} by Nakamoto Satoshi in 2008 for developing bitcoins has attracted many researchers to the research of blockchains.
However, the blockchain has problems in privacy level, increased transaction time and scalability.
In order to overcome them, new cryptocurrencies with a wide variety of advantages are developed.
There are also blockchains that handle physical rights as well as virtual currency (e.g., ownership of cars, copyrights of songs, proof of circulation and so on) \cite{3}.
It is a great advantage of blockchains that it allows us to exchange them in the absence of any trusted third parties.\par

As trading on blockchains becomes popular, demands for trading across multiple blockchains increase \cite{3}.
As a specific example of exchanging assets across multiple blockchains among multiple parties,
consider the case that Alice wants to sell the copyrights of her songs for bitcoins.
Bob is willing to buy her copyrights with alt-coins.
Carol wants to exchange alt-coins for bitcoins.
An \emph{atomic cross-chain swap} protocol is a mechanism by which multiple parties exchange their assets managed by multiple blockchains.
It is common that some parties do not know and do not trust each other, thus the protocol must guarantee that no party conforming to the protocol suffers from a loss in their trading.
Specifically,  this protocol guarantees the following three conditions.
(1) if all parties conform to the protocol, then all assets are exchanged among parties,
(2) even if some parties or coalitions of parties deviate from the protocol, no party conforming to the protocol suffers a loss, and
(3) no coalition has an incentive to deviate from the protocol.
The more blockchain users have request to trade as blockchain technology develops in the future, the more important atomic cross-chain swaps are \cite{10,11,12,13,14,15,16,17}. \par

In many blockchains, assets are transferred from one party to another party by using \emph{smart contracts}.
A smart contract is a program that runs on a blockchain and has its correct execution enforced by the consensus protocol \cite{4,5}.
In this paper, \emph{hashed timelock contracts} (HTLCs) \cite{6,7} are used.
In HTLC, recipients of a transaction have to acknowledge payment by generating cryptographic proof within a certain timeframe.
Otherwise, the transaction does not take place.
For example, consider the case Alice wants to send an asset to Bob by using HTLCs in gratitude for taking money from Bob.
Alice first generates random secret data $s$, called a \emph{secret}, and produces \emph{hashlock} $h = H(s)$, where $H$ is a cryptographic hash function.
Next,  Alice publishes the contract with hashlock $h$.
After that, if Alice takes money form Bob, Alice reveals the secret $s$ to Bob.
When Bob sends the secret $s$ to the contract, the contract irrevocably transfers Alice's asset to Bob.
Alice also sets \emph{timelock} $t$ so that her escrowed asset can be returned if Bob does not give money to Alice within the time. \par
In this paper, we consider the case that more than two parties exchange their assets, as shown in Figure \ref{fig6}.

\begin{figure}[htbp]
	\begin{center}
		\includegraphics[width= 0.6 \linewidth,clip]{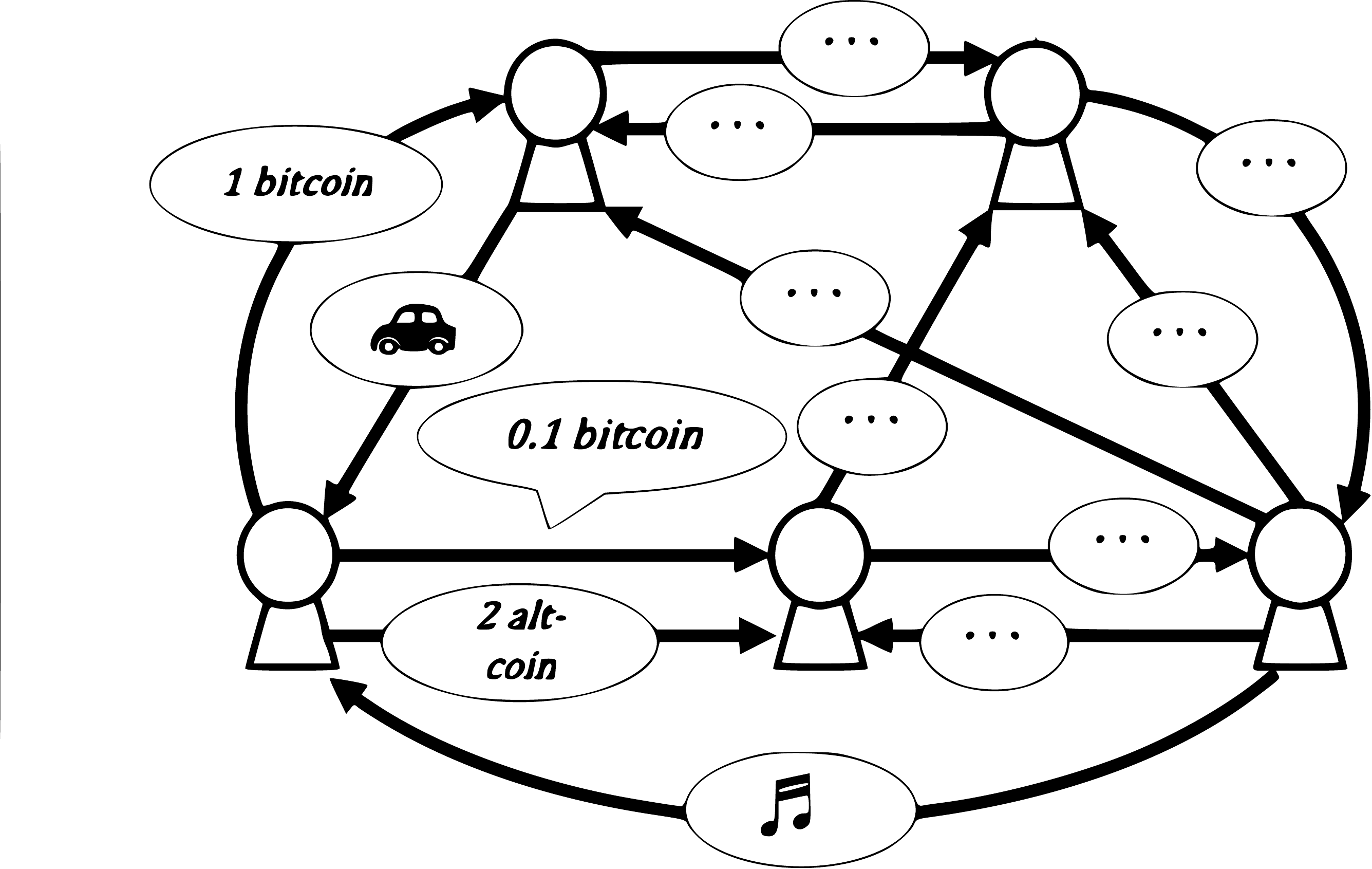}
		\caption{Several parties exchange their assets by using multiple blockchains. For example, one person get ownership of the car and a copyright of the song for 1.1 bitcoin and 2 alt-coin.}
		\label{fig6}
	\end{center}
\end{figure}

In the following, we quote a simple protocol presented by Herlihy \cite{1} for exchanging assets among Alice, Bob, and Carol, as illustrated in Figures \ref{fig:1} and \ref{fig:2}.

\begin{figure}[p]
  \begin{center}
    \includegraphics[clip,width=6.0cm]{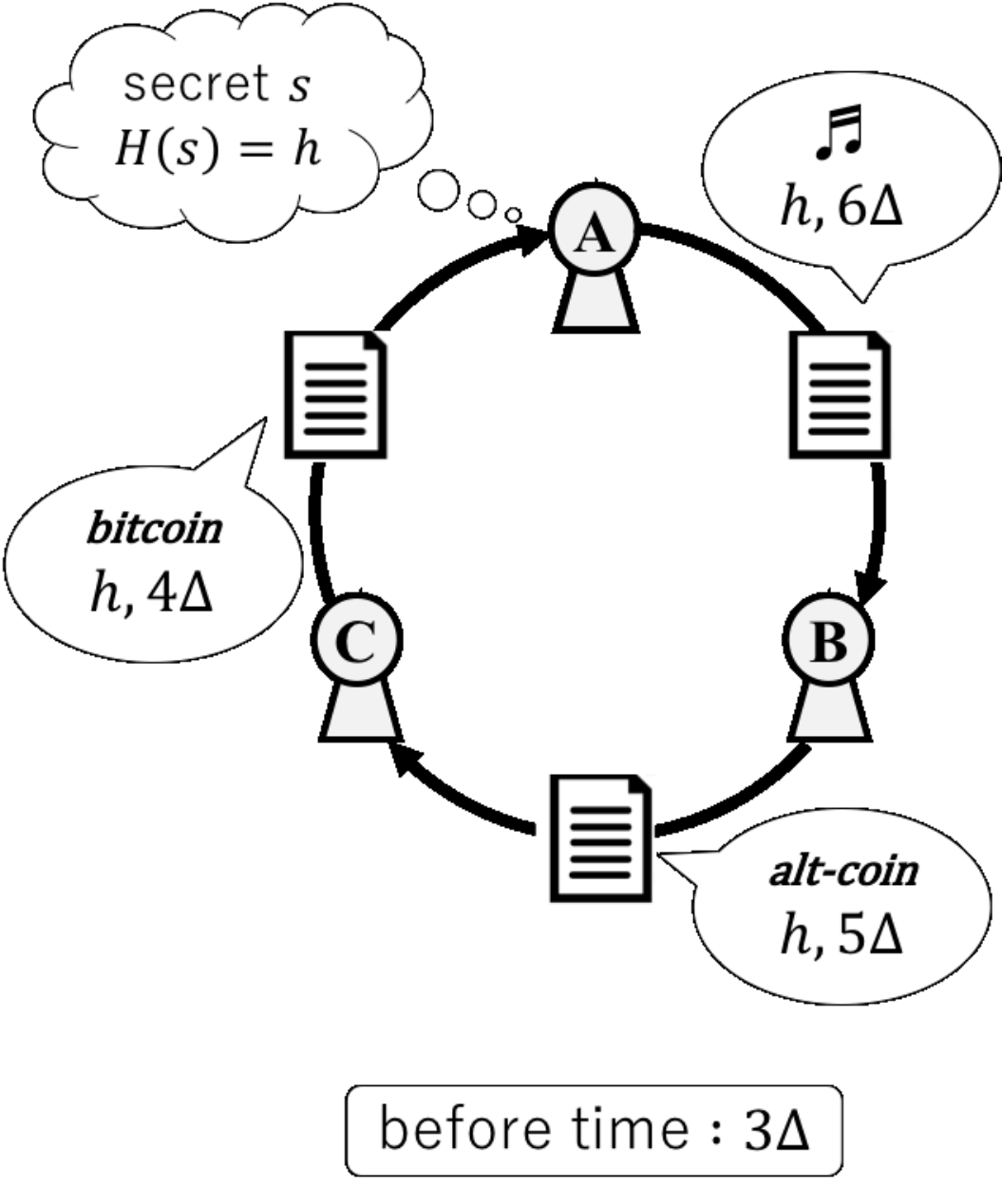}
    \caption{Alice, Bob and Carol publish smart contracts on blockchains}
    \label{fig:1}
  \end{center}
\end{figure}

\begin{figure}[htbp]
  \begin{center}
    \includegraphics[clip,width=15.0cm]{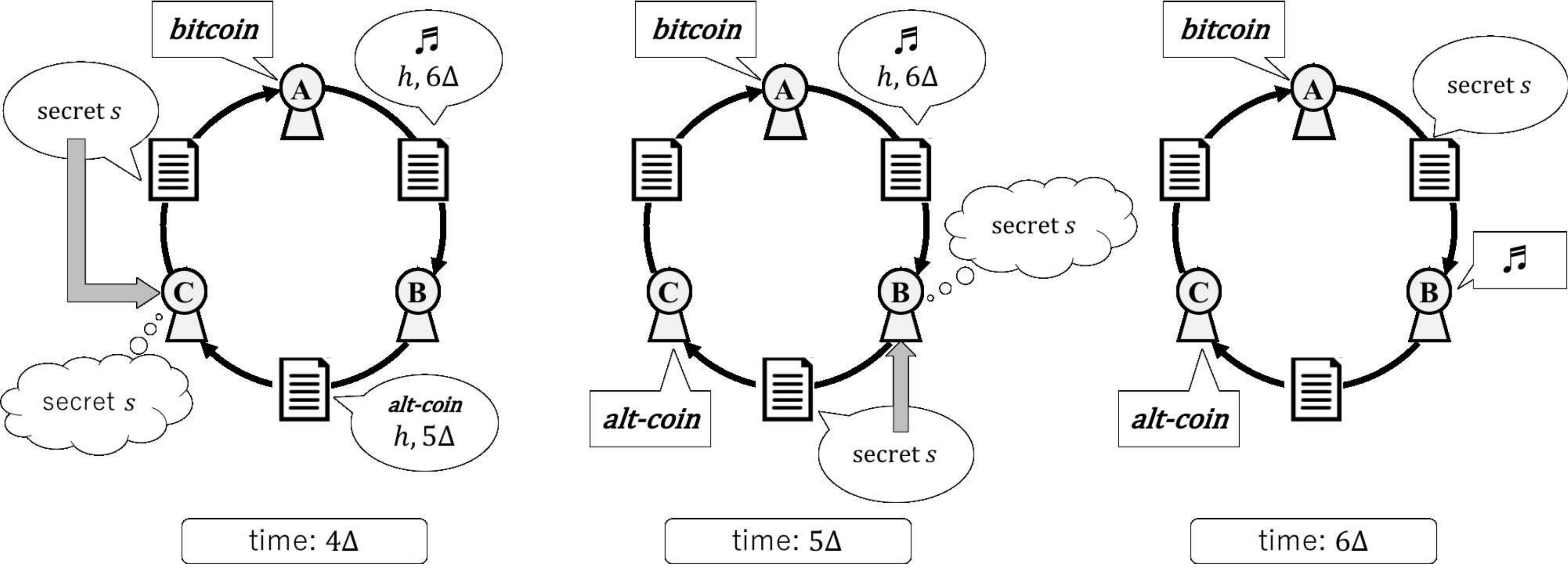}
    \caption{They acquire assets by sending secret $s$ to smart contracts in each time}
    \label{fig:2}
  \end{center}
\end{figure}

In that exchanging, Alice sends the copyrights of her songs to Bob, Bob sends alt-coins to Carol and Carol sends bitcoins to Alice.
Let $\Delta$ be time enough for one party to publish a smart contract on any of the blockchains,
or to change the state of a contract, and for the other party to detect the change.\par
\vspace{0.1in}
{\large $\cdot$ A simple protocol presented by Herlihy \cite{1}}
\begin{enumerate}
	\item Alice creates a secret $s$, hashlock $h = H(s)$, and publishes a contract on the music copyright blockchain with hashlock $h$ and timelock $6\Delta$ in the future, to transfer her music copyrights to Bob.
	\item When Bob confirms that Alice's contract has been published on the copyright blockchain, he publishes a contract on the alt-coin blockchain with the same hashlock $h$ but with timelock $5\Delta$ in the future, to transfer his alt-coins to Carol.
	\item When Carol confirms that Bob's contract has been published on the alt-coin blockchain, she publishes a contract on the Bitcoin blockchain with the same hashlock $h$, but with timelock $4\Delta$ in the future, to transfer her bitcoins to Alice.
	\item When Alice confirms that Carol's contract has been published on the Bitcoin blockchain, she sends the secret $s$ to Carol's contract, acquiring the bitcoins and revealing $s$ to Carol.
	\item Carol then sends $s$ to Bob's contract, acquiring the alt-coins and revealing s to Bob.
	\item Bob sends $s$ to Alice's contract, acquiring the copyrights and completing the swap.
\end{enumerate}

Everyone can stop the swap if published contracts are different from predetermined ones.
There is no possibility that Alice transfers to Bob the copyrights without acquiring the bitcoins because only Alice initially knows secret $s$.
If Carol's bitcoins have been transferred, this guarantees that she can get $s$ and acquire the alt-coins as well because she publishes her contract after confirming publication of Bob's contract and the timeout specified by the time lock of Bob's contract is one $\Delta$ greater than that of her contract.
Bob also acquires the copyrights if his alt-coins have transferred. 
Even if Alice and Bob conspire to deceive Carol, they can not get the bitcoins without payment for Carol, thus Carol never suffers from a loss.
As seen from these facts, every party should publish his contracts only after all the contracts for assets transferred to him are published unless he generates a secret for a swap.
We call a party who generates a secret for a swap a \emph{leader}, and denote the set of leaders by $L$.
From the above discussion, $L$ must be a feedback vertex set.
This simple protocol works if and only if we have exactly one leader $l$ such that  $\{l\}$ is a feedback set. \par

Generally, a cross-chain swap is modeled as a directed graph $D=(V,A)$. Vertex set $V$ denotes a set of parties and arc set $A$ denotes a set of proposed asset transfers.
When there are multiple cycles, this simple protocol does not work.
For example, there are three parties and all want to exchange each other in the swap, as illustrated in Figure \ref{fig:3}.

\begin{figure}[htbp]
  \begin{center}
    \includegraphics[clip,width=15.0cm]{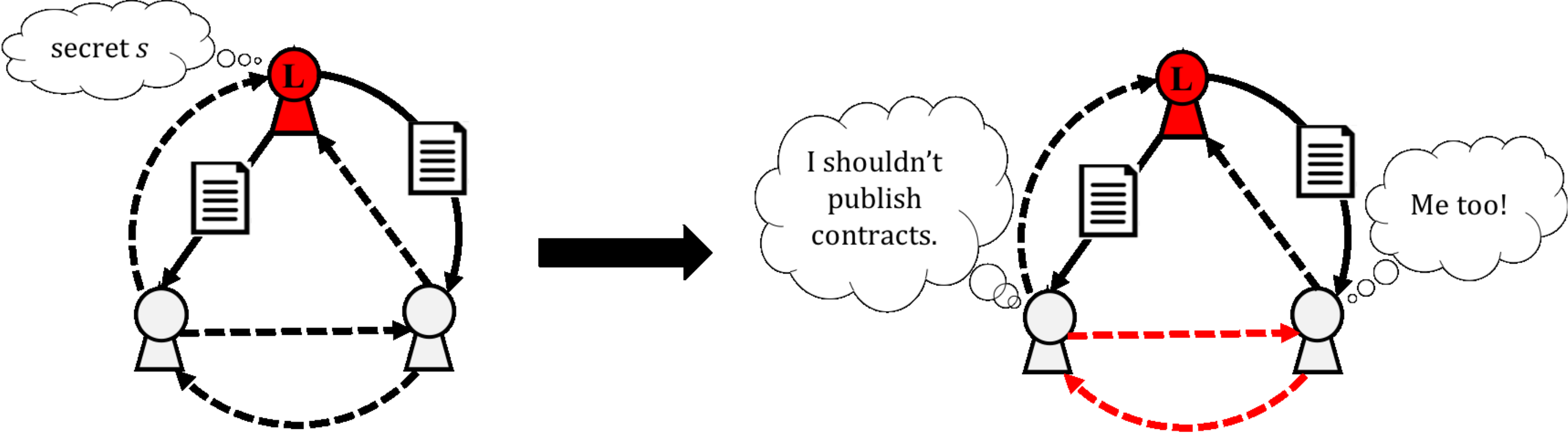}
    \caption{A case of only one leader in multiple cycles}
    \label{fig:3}
  \end{center}
\end{figure}

The leader can publish his two contracts without anxiety about a result of the swap. 
However, each of other two parties cannot publish his contracts until he confirms the other publishes a contract to him. 
Therefore, at least one leader is necessary on each cycle of a swap graph.
Moreover, if there are multiple cycles, it is not possible to assign consistent timeouts across cycles in a way that guarantees a gap of at least $\Delta$ between contracts for giving assets and those for taking assets.
That is, it is not easy to realize a mechanism such that each leader reveals its secret with confidence when there are multiple leaders.\par

To circumvent the problem, Herlihy \cite{1} sets complicated conditions for each smart contract to design an atomic cross-chain swap protocol. 
In the Herlihy's protocol, each leader creates a secret and one hashlock for each of the $|L|$ secrets is set in each contract in order to publish contracts safely.
Each leader sends his own secret with his signature to the contracts corresponding the asset which he wants to get.
The party who published the contract looks at it and gets the secret with the leader's signature.
Parties unlocks the hashlock created by the leader by sending all secrets with the signature and his own signature.
Each hashlock has its timelock, which depends on length of the path along which the secret is transferred from the leader generating the secret.
The sequence of signatures of parties indicates the path.
It is necessary to check whether the path actually exists in a given swap. 
Therefore, topology information of a swap graph is stored in every contract, which requires $O(|A|)$ bits.
A contract is triggered when all secrets of each leader are sent into its timelock to the contract and the network nodes of the blockchain of the corresponding contract verify that these paths exist.

\subsection{Our contributions}
Let $D(V,A)$ be a given directed graph that represents a cross-chain swap to general graphs.
We propose an atomic cross-chain swap protocol, which improves the space complexity (bits stored on all blockchains) of Herlihy's protocol \cite{1} from $O(|A|^2)$ to $O(|A|\cdot|V|)$, and the local time complexity (the maximum execution time of a contract in a swap to transfer the corresponding asset) from $O(|V|\cdot|L|)$ to $O(|V|)$.
The protocol dose not store the swap topology in a contract.
Instead, it simply assigns each contract the timelock depending on the number of signatures of parties.
Specifically, we set the timelock of a contract as follows: the contract is triggered if and only if the contract is given signatures of distinct $x$ parties in a swap before time $t+x\Delta$ for some $x \in [1,n]$ where $t$ is the time instance we will explain later (the starting time of the execution plus alpha).
Suppose that an asset of party $v$ is transferred to another party by executing a contract.
This means that the signature of $x$ parties are published in the corresponding blockchain before time $t+x\Delta$.
Then, party $v$ can trigger all the contracts that transfer an asset to $v$ because $v$ can obtain the signature of $x$ parties from that blockchain and can give those contracts $x+1$ signatures by adding its own signature before time $t+(x+1)\Delta$.
As a result, it is not necessary to store the topology information of a swap in any contract, thus the space complexity is improved.\par

Let \emph{local time complexity} be how long a contract takes to be triggered from executing a first action for triggering, which is sending a secret to the contract from a party at first.
The local time complexity is much smaller than the execution time of the protocol.
However, it is worthwhile to mention that each contract is executed by \emph{all} network nodes that host the copies of the ledger of the corresponding blockchain.
Therefore, the local time complexity is an index independent of the execution time and reducing local time complexity of a smart contract is of practical importance.
In Herlihy's protocol, a contract must verify at most $|V|$ signatures for the hashlock of each leader.
Thus, we need $\Theta(|V|\cdot |L|)$ time to trigger one contract.
On the other hand, in the proposed protocol, we only need to verify $|L|(\le |V|)$ secrets and at most $|V|$ signatures.
That is why the proposed protocol improves local time complexity from $|V|\cdot|L|$ to $|V|$.
We summarize our contribution to Table \ref{table:ref}.

\begin{table}[t]
\centering
\caption{Comparison between the performances of Herlihy's protocol and our protocol.}
\label{table:ref}
\begin{tabular}{|c | c | c | c |}
\hline
 & local time complexity & execution time & space complexity \\
\hline
Herlihy's protocol & $O(|V|\cdot|L|)$ & $O(|V|\Delta)$ & $O(|A|^2)$\\
Our protocol & $O(|V|)$ & $O(|V|\Delta)$ & $O(|A|\cdot|V|)$\\
\hline
\end{tabular}
\end{table}

\section{Model}\label{model}
\subsection{Graph Model}
$D=(V,A)$ is a \emph{directed graph}, where $V$ is a finite set of \emph{vertexes}, and $A$ is a finite set of ordered pairs of distinct vertexes called \emph{arcs}.
The number of vertexes $|V|$ is denoted by $n$.
An arc $(u,v)$ has tail $u$ and head $v$. 
An arc \emph{leaves} its tail and \emph{enters} its head.
Therefore, an arc $(u,v)$ is \emph{leaving arc} from $u$ and \emph{entering arc} to $v$.
A digraph $D'=(V',A')$ is a \emph{subdigraph} of $D=(V,A)$ if $V' \subseteq V$, $A' \subseteq A \cap (V' \times V')$ . \par

A \emph{path} $p$ in $D$ is a sequence of vertexes $(v_0,v_1,...,v_k)$ with $(v_i,v_{i+1}) \in A$ for every $i=0,1,\dots ,k-1$, where we define $k$ as the length of the path.
The \emph{distance} $dist(u,v)$ from vertex $u$ to vertex $v$ is the length of the shortest path from $u$ to $v$ in $D$.
$D$'s \emph{diameter} $diam(D)$ is the largest distance between two vertexes in $D$.
A vertex $v$ is \emph{reachable} from vertex $u$ if there is a path from $u$ to $v$.
Directed graph $D$ is \emph{strongly connected} if every vertex is reachable from any other vertex in $D$.
A path $(u_0,...,u_k)$ with $k \ge 1$ is a \emph{cycle} if $u_0 = u_k$.
A \emph{feedback vertex set} of $D$ is a set of vertexes whose removal makes a graph acyclic.
By definition, any feedback vertex set contains at least one vertex in any cycle in $D$.

\subsection{Blockchain and Smart Contract Model}\label{contract model}
A \emph{blockchain} is a distributed ledger that can record transactions between parties in a verifiable and permanent way.
We do not assume any specific blockchain algorithm.
We only assume that every blockchain used in a cross-chain swap supports \emph{smart contracts}.
An owner of an asset can transfer ownership of the asset to a \emph{counterparty} by creating a smart contract.
The owner specifies conditions in the contract to transfer the asset.
A counterparty can get ownership of the asset when the conditions are satisfied.
We say that a contract is \emph{published} when the owner of the corresponding asset releases the contract on a blockchain.
A contract is \emph{triggered} when all conditions of the contract are satisfied and the ownership of the asset on the contract is transferred to the counterparty.
The owner of an asset can also specify conditions to get back the asset on a smart contract.
If the conditions are satisfied before the contract is triggered, the asset is returned to the owner.
Thereafter, the asset is never transferred to the counterparty by this contract.
Typically, the owner can specify the condition to get back the asset in the form of \emph{timeout condition}.
The owner can regain the asset anytime after the specified time in a contract.
Smart contracts are immutable, which means smart contracts can never be changed and no one can tamper with or break the contract.
We assume that every operation on any blockchain can be completed within a known amount of time $\Delta$.
In particular, every party can publish a contract within $\Delta$ unit time and every party can trigger a contract within $\Delta$ unit time (if he has all information to satisfy all conditions of the contract).
Blockchains require each party $v$ to create the public key $pkey_v$ and the private key $skey_v$ for transactions.
A \emph{signature} $sign(x,v)$ is the the signature on data $x$ signed with $skey_v$.
Every party can verify the signature $sign(x,v)$ with the public key $pkey_v$.\par

\subsection{Swap model}\label{swap model}

Parties $v_1,v_2,\dots ,v_n$ are given a digraph $D=(V,A)$, where each vertex in $V$ represents a party, and each arc $(u,v) \in A$ represents an asset on a blockchain to be transferred from party $u$ to party $v$.
When ownership of an asset on an arc $(u,v)$ transfers from $u$ to $v$, we say the asset transfer happens, or the arc is triggered. \par

We denote the corresponding asset for arc $(u,v) \in A$ as $a_{u,v}$. 
The value of each asset may vary from one party to another; for example, $v$ may find a higher value on the asset than $u$.
Otherwise, that is, if all the parties see the same value about an asset, exchanging their assets following $D=(V,A)$ is meaningless (no one can gain unless someone loses.).
The value of $a_{u,v}$ for $u$ (resp. $v$) is denoted by $value^-_{u,v}$ (resp. $value^+_{u,v}$).
We define $value^+_{D,v}$ as the sum of the values for $v$ of all entering arcs to $v$ in $D$, i.e., $value^+_{D,v} = \sum_{(u,v) \in A } value^+_{u,v}$.
We also define $value^-_{v,D} = \sum_{(v,u) \in A} value^-_{v,u}$ as the sum of the values for $v$ of all leaving arcs from $v$ in $D$.
\par

We assume that $D$ meets the following condition because otherwise party $v$ should not participate in a swap $D$.
\begin{center}
	$\forall v \in V : value^+_{D,v} > value^-_{v,D}$
\end{center}

Moreover, we assume that there exists no connected subgraph $D'=(V',A)$ of $D$ such that
\begin{center}
	\begin{eqnarray}
	\exists C \subset V' : (\sum_{v' \in C} (value^+_{D',v'} - value^-_{v',D'})  > \sum_{v' \in C} (value^+_{D,v'} - value^-_{v',D})) \\
	\land  (\forall v'' \notin C : value^+_{D',v''} - value^-_{v'',D'}  \ge  value^+_{D,v''} - value^-_{v'',D}).
	\end{eqnarray}
\end{center}

Inequality (1) means that one coalition $C$ in $D'$ gets more benefits in $D'$ than $D$, and
(2) means that any party in $D'-C$ gets equal to or more benefits in $D'$ than $D$.
If swap $D$ has such a subgraph $D'$, we can say that $D$ forces $C$ to perform disadvantage exchanges because $C$ gets larger benefits in swap $D'$.
Therefore, $D$ is not an appropriate swap.
This is the reason why we exclude such a swap.
As we will explain in Section \ref{discussion}, this assumption is weaker than the assumption that Herlihy \cite{1} makes to guarantee that no party and coalition can get more benefits by deviating from the protocol than following the protocol.
Even without this assumption, the proposed protocol is an atomic cross-chain swap protocol under another specific assumption about value function which Herlihy's protocol requires.
We discuss that in the Section \ref{discussion}.\par

A protocol is a strategy for a party, that is, a set of rules that determines which action the party takes.
Ideally, all parties in a swap $D=(V,A)$ follows the common protocol $P$.
However, we must consider the case that some party does not follow (i.e, deviates from) the common protocol to get larger benefits.
To make the matter worse, those parties may make a coalition and take actions cooperatively to get larger benefits in total (some party in the coalition willingly loses aiming at larger total benefits of the coalition) .
We must design a protocol such that any party following the protocol does not suffer from a loss even if such selfish parties or coalitions exist.
We assume that every party in a swap $D=(V,A)$ can send any message to any party in the swap.
\emph{Space complexity} of a protocol is measured by the total number of bits required store information on all blockchains in a swap.
\emph{Local time complexity} of a protocol is measured by the maximum execution time of a contract in a swap to transfer the corresponding asset.

\begin{dfn}
	A swap protocol $P$ is \emph{uniform} if it guarantees the followings:
	\begin{itemize}
		\item If all parties follow $P$, all arcs in $D=(V,A)$ are triggered.
		\item Even if there are parties arbitrarily deviating from $P$, every party $v$ following $P$ gets all assets of entering arcs to $v$ or regains all assets of leaving arcs from $v$.
	\end{itemize}
\end{dfn}

\begin{dfn}
	A swap protocol $P$ is \emph{Nash equilibrium} if no party can get more benefits by deviating from $P$ than following $P$ unless he joins a coalition.
\end{dfn}

\begin{dfn}
	A swap protocol $P$ is \emph{strong Nash equilibrium} if it guarantees that no party and coalition can get more benefits by deviating from $P$ than following $P$.
\end{dfn}

\begin{dfn}
	A swap protocol $P$ is  \emph{atomic} if it guarantees to be uniform and strong Nash equilibrium.
\end{dfn}

\section{Proposed protocol}\label{protocol}
\subsection{Outline of proposed protocol}
\label{proposed protocol}
The proposed protocol $P$ consists of four phases.
In Phase 1, every party finds a common feedback vertex set $L$ of $D = (V, A)$ locally by using the same algorithm.
Although finding a minimum feedback vertex set is NP-hard \cite{8}, we do not need the minimum set, thus we can use any approximate solution.
For example, there exists an algorithm to find a 2-approximate solution \cite{9}.
The parties belonging to the vertex set $L = \{l_1,l_2,\dots, l_k\} $ are called leaders.
We call the other parties $f_1, f_2, \dots, f_{n-k}$ followers.
In this phase, every leader generates a secret which is a random bit string and calculates a hashlock based on the secret.
In Phase 2, a smart contract corresponding to each arc of $D$ is published.
Each leader spontaneously publishes contracts for all his leaving arcs.
Each follower publishes contracts for all his leaving arcs after confirming that the contracts for all his entering arcs are already published.
In Phase 3, each of the leaders $l_2, l_3, \dots, l_k$ sends its secret to $l_1$ after confirming that the contracts for all its entering arcs have been published.
The leader $l_1$ starts Phase 4 after $l_1$ confirms that the contracts for all the entering arcs to $l_1$ have been published and all secret have been sent from each other leaders.
In Phase 4, each arc of $D$ is sequentially triggered, which starts from leader $l_1$.\par

As described in Section 3.2, we design the contracts in Phase 2 so as to guarantee the following three properties.
(i) For any arc, no party can trigger the corresponding contract without knowing the secrets of all leaders.
(ii) For any party $v$, if a contract on leaving arc from $v$ is triggered, $v$ can trigger the contracts published on all entering arcs to $v$. 
(iii) For any arc $(u,v)$, party $v$ can regain $a_{u,v}$ if the contract on $(u,v)$ are not triggered during a certain period after it was published.
If a leader $l_i$ follows the proposed protocol, $l_i$ sends its secret to the leader $l_1$ after confirming that the contracts on all entering arcs have been published.
Therefore, from Property (i) and (ii), it is guaranteed that no leaving arc from $l_i$ is triggered unless all entering arcs to $l_i$ are triggered.
Moreover, from property (iii), after a certain period of time, $l_i$ gets assets of all entering arcs to $l_i$, or regains assets of all leaving arcs from $l_i$.
If a follower $f_i$ follows the proposed protocol, $f_i$ publishes the contracts for all leaving arcs after confirming that the contracts for all entering arcs have been published.
Therefore, according to the same argument, after a certain period of time, $f_i$ gets assets of all entering arcs to $f_i$, or regains assets of all leaving arcs from $f_i$.\par

In the following, we describe the trigger condition and regain condition of the smart contract in Section 3.2, and the detailed operation in Phases 1, 2, 3 and 4 in Sections 3.3, 3.4, 3.5 and 3.6, respectively.
We describe codes for a smart contract, and the behaviors of the the leader $l_1$, the other leaders $l_2,l_3,\dots,l_k$ and followers in Algorithm 1, 2, 3, 4.
We distinguish entering arcs to each party $v$ by locally labeling those with $(1,v),(2,v),...,(p,v)$ where $p$ is the number of the entering arcs.
We also label leaving arcs from each party $v$ with $(v,1),(v,2),...,(v,q)$ where $q$ is the number of the leaving arcs.
In the codes, we abbreviate Smart Contract to SC.

\subsection{The conditions of smart contracts}\label{condition}
In this subsection, we specify what each party writes in a contract.
In a contract, its publisher specifies the conditions for transferring and regaining the asset.
Each leader $l_i$ makes hashlock $H(s_i)=h_i$ with random secret data $s_i$, called a secret, where $H()$ is a cryptographic hash function common to all parties.
In any blockchain, as mentioned in Section \ref{contract model}, time $\Delta$ is enough to publish and trigger a contract.   
The condition of the contract on arc $(u,v)$ to transfer asset $a_{u,v}$ is as follows:
Asset $a_{u,v}$ is transferred to counterparty $v$ if $v$ sends the secrets $s_1,s_2,\dots,s_k$ of all leaders and $x$ signatures \footnote{$sign((s_1,s_2,\dots ,s_k),v_1),sign((s_1,s_2,\dots ,s_k),v_2),\dots,sign((s_1,s_2,\dots ,s_k),v_x))$} on k-tuple $(s_1,s_2,\dots,s_k)$ by arbitrary $x$ parties to the contract before time $t_s + ((diam(D) + 1)\Delta + 2\varepsilon) + x\Delta$, where $t_s$ is the starting time of the protocol.
That is, the more a party collects the signatures of the parties, the later he can trigger the contracts on his entering arcs.
Let $\varepsilon(<< \Delta)$ be the time required to complete each of Phases 1 and 3.
We set the deadline of the smart contract to time $t_s + ((diam(D)+1)\Delta + \varepsilon)+n\Delta$, so that the publisher can regain his asset on the contract if the corresponding contract is not triggered before that time.

\begin{algorithm}[H]
\caption{Contents of smart contract}
\begin{algorithmic}[1]
\label{algo1}
\STATE  \textsc \bf{Swap Contract}
\STATE \Int $h_1,h_2,\dots ,h_k$;
\STATE \Int $pkey_1,pkey_2,\dots ,pkey_n$
\STATE address party;
\STATE address counterparty;
\STATE \Uint start = protocol starting time;
\STATE
\STATE function transfer(\Int $s_1,s_2,\dots ,s_k$, \Int $sig_1,sig_2,\dots ,sig_k$) 
\STATE \tab require (msg.sender == counterparty);
\STATE \tab \If \ $H(s_i)$ = $h_i$ holds for all $i=1,2,\dots ,k$ and \\
there exist $x$ distinct parties $u_1,u_2,\dots ,u_x$ in $V$ such that $sig_1,sig_2, \dots ,sig_x$ are correct signature on $s_1,s_2,\dots ,s_k$ by $u_1,u_2,\dots ,u_x$ respectively and \\
$now < start + (diam(D)+x+1)\Delta + \varepsilon$
\STATE \tab \tab \Then transfer asset to counterparty;
\STATE \tab \tab halt;
\STATE \tab \End \If
\STATE 
\STATE function timeout()
\STATE \tab require (msg.sender == party);
\STATE \tab \If \ $now > start + (diam(D)+n+1)\Delta + \varepsilon$
\STATE \tab \tab \Then return asset to party;
\STATE \tab \tab \tab halt;
\STATE \tab \End \If
\end{algorithmic}
\end{algorithm}

\begin{algorithm}[H]
\caption{Behavior of top-leader $l_1$}
\begin{algorithmic}[1]
\label{algo1}
\STATE  \textsc \bf{Phase 1}
\STATE determine a feedback vertexes set ($l_1,l_2,\dots ,l_k$);
\STATE create secret $s_1$ and hashlock $h_1 = H(s_1)$;
\STATE broadcast $h_1$ to all parties (including $l_1$ itself);
\STATE broadcast $pkey_{l_1}$ to all parties (including $l_1$ itself);
\STATE receive each hashlock $h_i$ from each leader $l_i$ with $1 \le i \le k$;
\STATE receive each public key $pkey_i$ from each party $v_i$ with $1\le i \le n$;
\STATE 
\STATE  \textsc \bf{Phase 2}
\STATE create and publish SC for all leaving arcs from $l_1$;
\STATE
\STATE \textsc \bf{Phase 3}
\STATE receive each secret $s_i$ from each leader $l_i$ with $2 \le i \le k$;
\STATE \If \ $\exists i \in \{2,3,\dots,k\}: H(s_i) \neq h_i or \exists j \in \{1,2,\dots,p\}$: a consistent contract is not published on arc $(j,l_1)$  or $now \ge start + (diam(D)+1)\Delta + 2\varepsilon$
\STATE \tab \Then \Goto \ \textsc{Regain};
\STATE \End \If
\STATE \If \ $\exists i \in \{2,3,\dots,p\}: (i,l_1)$ don't have SC 
\STATE \tab \Then \Goto \ \textsc{Regain}; 
\STATE \End \If
\STATE
\STATE \textsc \bf{Phase 4}
\STATE $sig_{l_1} \gets$ my signature on $(s_1,s_2,\dots ,s_k)$;
\STATE \For \ $i = 1$ to $p$
\STATE \tab call function transfer($(s_1,s_2,\dots ,s_k)$,$(sig_{l_1})$) for the contract on $(i,l_1)$;
\STATE \End \For
\STATE halt;
\STATE
\STATE \textsc \bf{Regain}
\STATE \Wait \ until $now > start + (diam(D)+n+1)\Delta + 2\varepsilon$;
\STATE call function timeout() for the contracts on all non-triggered leaving arcs from $l_1$;
\end{algorithmic}
\end{algorithm}

\begin{algorithm}[H]
\caption{Behavior of sub-leader $l_j$}
\begin{algorithmic}[1]
\label{algo1}
\STATE  \textsc \bf{Phase 1}
\STATE determine a feedback vertexes set ($l_1,l_2,\dots ,l_k$);
\STATE create secret $s_j$ and hashlock $h_j = H(s_j)$;
\STATE broadcast $h_j$ to all parties (including $l_j$ itself);
\STATE broadcast $pkey_{l_j}$ to all parties (including $l_j$ itself);
\STATE receive each hashlock $h_i$ from each leader $l_i$ with $1 \le i \le k$;
\STATE receive each public key $pkey_i$ from each party $v_i$ with $1\le i \le n$;
\STATE 
\STATE  \textsc \bf{Phase 2}
\STATE create and publish SC for all leaving arcs from $l_j$;
\STATE
\STATE \textsc \bf{Phase 3}
\STATE \Wait \ until $\forall i \in \{1,2,\dots,p\}$: a consistent contract is published on arc $(i,l_j)$;
\STATE \If \ $now > start + (diam(D)+1)\Delta + 2\varepsilon$;
\STATE \tab \Then \Goto \ \textsc{Regain}; 
\STATE \End \If
\STATE send $s_j$ to top-leader $l_1$;
\STATE
\STATE \textsc \bf{Phase 4}
\STATE \Wait \ until any leaving arc is triggered or $now > start + (diam(D)+n+1)\Delta + 2\varepsilon$; 
\STATE \If \ $now > start + (diam(D)+n+1)\Delta + 2\varepsilon$
\STATE \tab \Then \Goto \ \textsc{Regain};
\STATE \tab \Else get $s_1,s_2,\dots ,s_k$ and $sig_1,sig_2,\dots ,sig_x$ from the triggered arc;
\STATE \tab $sig_{l_j} \gets$ my signature on $(s_1,s_2,\dots ,s_k)$;
\STATE \tab \For \ $i = 1$ to $p$
\STATE \tab \tab call function transfer($(s_1,s_2,\dots ,s_k)$,$(sig_1,sig_2,\dots ,sig_x and sig_{l_j})$) for the contract on $(i,l_j)$;
\STATE \tab \End \For
\STATE \tab halt;
\STATE \End \If
\STATE
\STATE \textsc \bf{Regain}
\STATE \Wait \ until $now > start + (diam(D)+n+1)\Delta + 2\varepsilon$;
\STATE call function timeout() for the contracts on all non-triggered leaving arcs from $l_j$;
\end{algorithmic}
\end{algorithm}

\begin{algorithm}[H]
\caption{Behavior of follower $f$}
\begin{algorithmic}[1]
\label{algo1}
\STATE  \textsc \bf{Phase 1}
\STATE determine a feedback vertexes set ($l_1,l_2,\dots ,l_k$);
\STATE broadcast $pkey_f$ to all parties (including $l_j$ itself);
\STATE receive each hashlock $h_i$ from each leader $l_i$ with $1 \le i \le k$;
\STATE receive each public key $pkey_f$ from each party $v_i$ with $1\le i \le n$;
\STATE 
\STATE  \textsc \bf{Phase 2}
\STATE \Wait \ until all entering arcs have SC or $now > start + (diam(D)+1)\Delta + \varepsilon$; 
\STATE \If \ $now > start + (diam(D)+1)\Delta + \varepsilon$
\STATE \tab \Then \Goto \ \textsc{Regain};
\STATE \tab \Else create and publish SC for all leaving arcs from $f$;
\STATE \End \If
\STATE
\STATE \textsc \bf{Phase 3}
\STATE 
\STATE \textsc \bf{Phase 4}
\STATE \Wait \ until any leaving arc is triggered or $now > start + (diam(D)+n+1)\Delta + 2\varepsilon$; 
\STATE \If \ $now > start + (diam(D)+n+1)\Delta + 2\varepsilon$
\STATE \tab \Then \Goto \ \textsc{Regain};
\STATE \tab \Else get $s_1,s_2,\dots ,s_k$ and $sig_1,sig_2,\dots ,sig_x$ from the triggered arc;
\STATE \tab $sig_f \gets $ my signature on $(s_1,s_2,\dots ,s_k)$;
\STATE \tab \For \ $i = 1$ to $p$
\STATE \tab \tab call function transfer($(s_1,s_2,\dots ,s_k)$,$(sig_1,sig_2,\dots ,sig_x and sig_f)$) for the contract on $(i,f)$;
\STATE \tab \End \For
\STATE \tab halt;
\STATE \End \If
\STATE
\STATE \textsc \bf{Regain}
\STATE \Wait \ until $now > start + (diam(D)+n+1)\Delta + 2\varepsilon$;
\STATE call function timeout() for the contracts on all non-triggered leaving arcs from $f$;
\end{algorithmic}
\end{algorithm}

\newpage
\subsection{Phase 1 : Preparation}
Every party takes $D=(V,A)$ as input.
At first, he chooses leaders $l_1,l_2,\dots ,l_k$ such that the leaders form a feedback vertex set of $D$.
Since all parties are given the same swap $D=(V,A)$ as input, they can independently find the same leaders $l_1,l_2,\dots ,l_k$ without communication.
We call $l_1$ the \emph{top leader} and the other leaders $l_2,l_3,...,l_k$ \emph{sub-leaders}.\par
Next, each leader $l_i$ generates a sequence of random bits $s_i$ called a secret and computes hashlock $h_i = H(s_i)$, after which it sends only the hashlock to all parties.
Finally, all parties send their public keys $pkey_{v_1},pkey_{v_2},...,pkey_{v_n}$ for verifying their signatures.
We assume that all of these can be done in time $\varepsilon \ll \Delta$.

\subsection{Phase 2 : Publication}
Every leader spontaneously publishes contracts for all their leaving arcs with the conditions described in Section \ref{condition}.\par 

When a follower finds that the contracts on all his entering arcs are already published, he checks whether the contents of the contracts are consistent with the conditions of Section \ref{condition}.
Especially, he checks whether the public keys of all the parties and the hashlocks of all leaders that he receives in Phase 1 match the public keys of all parties and the hashlocks of all leaders specified in the contracts on entering arcs.
He quits the swap without publishing any contract if those public keys or hashlocks do not match.
Thus, even if the party deviating from $P$ may send a fake hashlock or a public key to the other parties in Phase 1, no other party following $P$ suffers from a loss.
If all the published contracts are consistent, he publishes contracts for all his leaving arcs.
As will be described later, each party $v$ reveals the signature on $s_1,s_2,\dots ,s_k$ with $v$'s secret key $skey_v$ only when $v$ triggers all entering arcs.
Therefore, all entering arcs of $v$ can be triggered if any leaving arc of $v$ is triggered.
Hence, no follower suffers from a loss, that is, every follower gets assets of all his entering arcs or regains assets of all his leaving arcs.\par
As described in Section \ref{proof}, Phase 2 can be completed within at most $(diam(D)+1)\Delta$ time.

\subsection{Phase 3 : Share secrets}
Every sub-leader reveals his secret to the top-leader if the contracts are published on all his entering arcs.
This is to ensure that the top-leader can trigger contracts at first in Phase 4.

The top-leader confirms whether or not the secrets acquired from the sub-leaders are correct using hashlocks shared at Phase 1.
If the top-leader finds an incorrect secret, it quits the swap without going to Phase 4.

We assume that Phase 3 can be done in time $\varepsilon \ll \Delta$.

\subsection{Phase 4 : Trigger}
The top-leader starts Phase 4 if he acquires the secrets of all the sub-leaders within time $(diam(D)+1)\Delta + 2\varepsilon$ which is enough to complete Phases 1, 2 and 3.
This is because the top-leader only need to send secrets of all leaders $s_1,s_2,...,s_k$ and its signature within the time $t_s + ((diam(D)+1)\Delta + 2\varepsilon)+\Delta$ to all his entering arcs in order to get the assets of all his entering arcs.
The top-leader sends these secrets and his signature to all his entering arcs, by which he triggers the contracts and acquires all assets of the entering arcs.\par

If the top-leader deviates from $P$, he may not trigger some (or all) contracts that transfer assets to the top leader.
However, those actions are irrational. 
We explain the reason in Lemma \ref{lem8}, (Section 4).\par

Next, we describe the behavior of the sub-leaders and followers in Phase 4.
Each party $v$ of them waits until any of his leaving arcs is triggered.
Consider that one of his leaving arcs is triggered with the information of the secrets $s_1, s_2, \dots, s_k$ and the signature of $x$ distinct parties $(1 \le x < n)$ on k-tuple ($s_1,s_2,\dots,s_k$).
By definition of contracts, the leaving arc must be triggered before $t_s + (diam(D)+1)\Delta+2\epsilon+x\Delta$. 
Then, party $v$ acquires from the contract the secrets $s_1, s_2, \dots, s_k$ and those $x$ signatures within time $t_s + (diam(D) + 1)\Delta + 2\varepsilon +x\Delta$.
He immediately sends $s_1, s_2, \dots, s_k$ and the $x$ signatures in addition to his signature (thus $x+1$ signatures in total) to all entering arcs.
As a result, $v$ sends all necessary information to the contracts on all his entering arcs before time $t_s+ ((diam(D)+1)\Delta + 2\varepsilon + (x+1)\Delta$, which guarantees that all his entering arcs are triggered.
We describe these in Figure \ref{fig:4}.\par

\begin{figure}[htbp]
	\begin{center}
		\includegraphics[width= 0.6 \linewidth,clip]{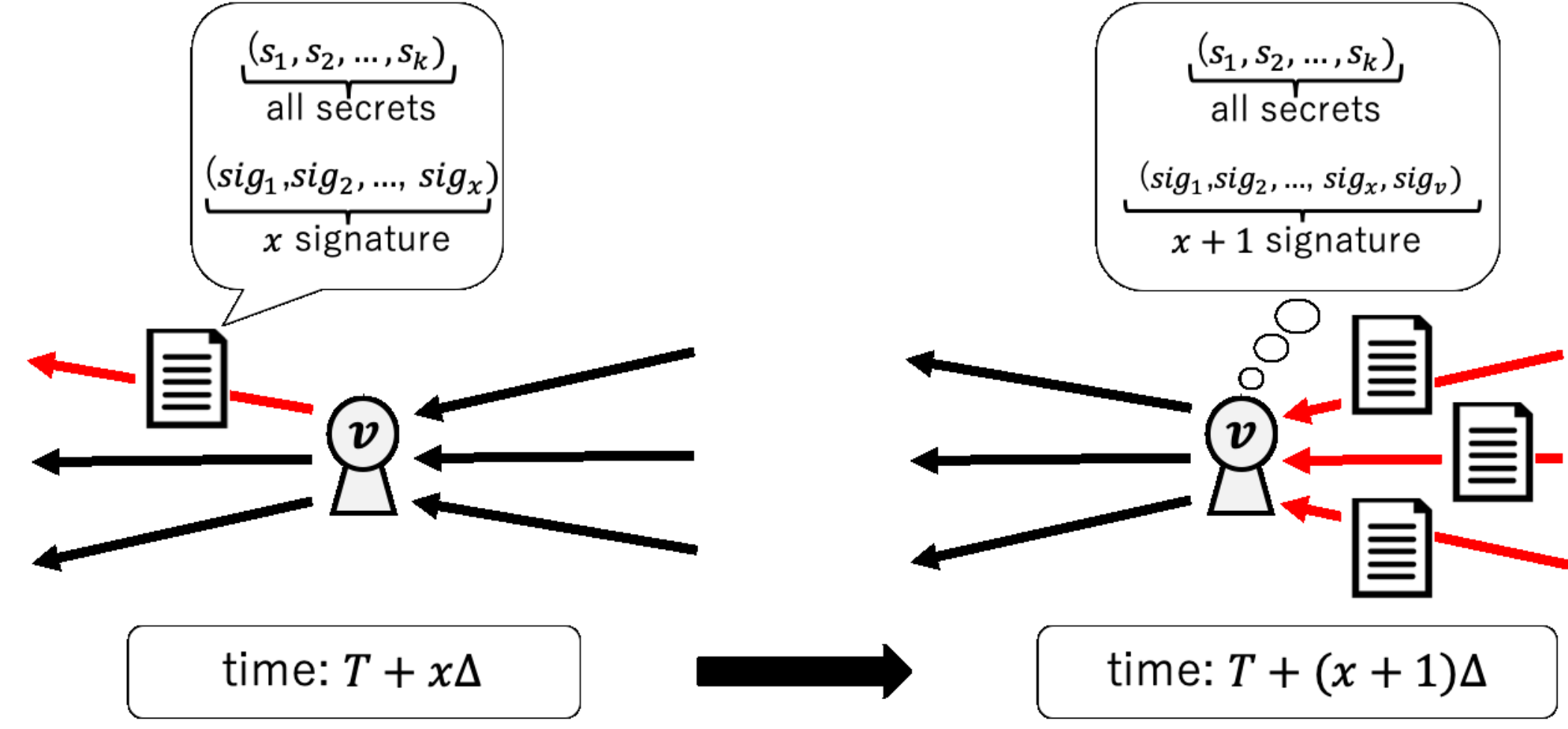}
		\caption{The party $v$ can trigger all his entering arcs when a leaving arc is triggered.}
		\label{fig:4}
	\end{center}
\end{figure}

Every party regains the asset of each of leaving arcs if it is not triggered by the deadline (specified by the timeout) $t_s + (diam(D)+n+1)\Delta + 2\varepsilon$.\par

\section{Correctness and Complexity of Protocol}\label{proof}
We prove that proposed protocol $P$ is atomic (i.e., uniform and strong Nash equilibrium).
First, we prove that $P$ is uniform.\par

\begin{lemma}\label{lem1}
	Assume that top-leader $l_1$ follows $P$. If any leaving arc of $l_1$ is triggered, all of entering arcs to $l_1$ are triggered.
\end{lemma}
\begin{proof}
	Assume by contradiction that there is an execution of $P$ such that a leaving arc $(l_1,v)$ of the top-leader $l_1$ is triggered and an entering arc $(u,l_1)$ is not triggered.
	After $(l_1,v)$ is triggered, $l_1$ immediately obtains secrets of all leaders, which is before the time $t_s + diam(D)+1)\Delta+2\varepsilon$.
	Thus, $l_1$ sends all the secrets and its own signature to the contract on $(u,l_1)$ before the time $t_s + diam(D)+1)\Delta+2\varepsilon$, which means that $l_1$ triggers $(u,l_1)$ in the execution.
	This contradicts the assumption.
\end{proof}

\begin{lemma}\label{lem2}
	Assume that sub-leader $l_i (i \neq 1)$ follows $P$. If any leaving arc of $l_i$ is triggered, all of entering arcs to $l_i$ are triggered.
\end{lemma}
\begin{proof}
	Assume for contradiction that there is an execution of $P$ such that a leaving arc $(l_i,v)$ of sub-leader $l_i$ is triggered but an entering arc $(u,l_i)$ is not triggered.
	All secrets $s_1,s_2,...,s_k$ and the signature of $x(<n)$ parties are revealed in the contract of $(l_i,v)$ by time $ t_s + (diam(D)+x+1)\Delta + 2\varepsilon$.
	By the assumption, those $x$ signatures do not include the signature of $l_i$ because $l_i$ reveals its signature only when it triggers the contracts of all entering arcs.
	The contract on $(l_i,v)$ requires the secret of all leaders and $l_i$ does not reveal its secret $s_i$ until the contracts of all its entering arcs are published.
	Therefore, the contracts of all entering arcs of $l_i$ are published before time $t_s + (diam(D)+x+1)\Delta + 2\varepsilon$.
	From the above, $l_i$ sends the secrets $s_1,s_2,\dots ,s_k$ and the $x$ signatures in addition to its signature on $s_1,s_2,\dots ,s_k$ to the contract on $(u,l_i)$ by time $t_s + (diam(D)+x+1)\Delta + 2\varepsilon$.
	This implies that edge $(u,l_i)$ is triggered in the execution.
	This is a contradiction.
\end{proof}

\begin{lemma}\label{lem3}
	Assume that follower $f$ follows $P$. If any leaving arc of $f$ is triggered, all of entering arcs to $f$ are triggered.
\end{lemma}
\begin{proof}
	Follower $f$ publishes a contract for a leaving arc from $f$ only after the contracts of all entering arcs to $f$ are published.
	Therefore, all entering arcs to $f$ are already published when any leaving arc from $f$ is triggered.
	We can prove the lemma in the same way as the proof of Lemma \ref{lem2}.
\end{proof}

\begin{lemma}\label{lem4}
	Assume that party $v$ follows $P$. If any leaving arc of $v$ is triggered, all of entering arcs to $v$ are triggered.
\end{lemma}
\begin{proof}
	The lemma immediately follows from Lemmas \ref{lem1}, \ref{lem2}, and \ref{lem3}
\end{proof}

\begin{lemma}\label{lem5}
	Let SC be the smart contract published on any leaving arc from party $v$ that follows $P$.
	It happens before time $t_s + (diam(D)+n+2)\Delta + 2\epsilon$ that contract SC is triggered or the asset of SC is regained by $v$.
\end{lemma}
\begin{proof}
	If a leaving arc from $v$ is not triggered before time $t_s + (diam(D)+n+1)\Delta + 2\varepsilon$, party $v$ calls the timeout function of the contract on the arc within at most $\Delta$ unit time after the deadline $t_s + (diam(D)+n+1)\Delta + 2\varepsilon$.
\end{proof}

\begin{lemma}\label{lem6}
	If every party follows $P$, all arcs are triggered within the time $t_s + 2(diam(D)+1)\Delta + 2\varepsilon$.
\end{lemma}
\begin{proof}
	If all parties follow $P$, each leader $l$ immediately publishes contracts for all leaving arcs from $l$ after the start of Phase 2, and each follower $f$ publishes contracts for all leaving arcs from $f$ immediately after the contracts for all entering arcs to $f$ are published.
	Since the time required to publish each contract is at most $\Delta$, we can see that all parties complete Phase 2 in $(diam(D) +1)\Delta$ time.
	Therefore, since the upper bound $\varepsilon$ of the time is required for the local computation and the message transmission of Phases 1 and 3, Phase 3 is completed by time $t_s + (diam(D)+1) \Delta + 2\varepsilon$.
	Thus, top-leader $l_1$ will immediately begin Phase 4.
	
	After Phase 4 has started, the leader $l_1$ immediately triggers the contracts of all entering arcs to $l_1$.
	Each party $v$ other than $l_1$ immediately triggers the contracts of all entering arcs to $v$, once a contract on any leaving arc from $v$ has been triggered.
	The signature on ($s_1,s_2,\dots,s_k$) by $v$ is never used to trigger the contract on the arc which is triggered for the first time among all leaving arcs from $v$.
	Therefore, $v$ triggers all entering arcs to $v$ by using the signatures used to trigger the leaving arc and his own signatures.
	Since the time required to trigger each contract is at most $\Delta$,  we can show, by the induction on distance from top-leader $l_1$, that the time required to trigger the contracts on all arcs in a given swap in Phase 4 is at most $(diam(D) + 1)\Delta$.
	Summing up these time for Phases 1 to 4, we can show that all arcs are triggered by time $t_s + 2(diam(D)+1)\Delta + 2\varepsilon$.
\end{proof}

\begin{lemma}\label{lem7}
	Protocol $P$ is uniform.
\end{lemma}
\begin{proof}
	The lemma immediately follows from Lemmas \ref{lem4}, \ref{lem5}, and \ref{lem6}
\end{proof}

\begin{lemma}\label{lem8}
	Protocol $P$ is strong Nash equilibrium.
\end{lemma}
\begin{proof}
	Assume by contradiction that, in an execution of $P$, a coalition $C$ formed from some parties gets more profits by deviating from protocol $P$ when the other parties follow $P$.
	That is, for $C$, the profits gained by $D'$ is bigger than that $C$ gets with $D$, where $D'=(V',A')$ is the subgraph of $D$ induced by all arcs triggered in the execution.
	In other words, the following inequality holds: \\
	\begin{center}
		$\sum_{v \in C} (value^+_{D',v} - value^-_{v,D'})  > \sum_{v \in C} (value^+_{D,v} - value^-_{v,D})$
	\end{center}
	
	Every party in $V' \setminus C$ triggers all his entering arcs by Lemma \ref{lem4}.
	Therefore, all entering arcs to $v'$ in $V' \setminus C$ are included in $A'$.
	Therefore, the following inequality holds for every party $v' \in V' \setminus C$.\\
	
	\begin{center}
		$\forall v' \in V' \setminus C: value^+_{D',v'} - value^-_{v',D'} =value^+_{D,v'} - value^-_{v',D'} \ge  value^+_{D,v'} - value^-_{v',D}$
	\end{center}
	However, these condition contradicts the assumption on a swap $D=(V,A)$ introduced in Section \ref{swap model}.\par
\end{proof}

\begin{thm}
	Protocol $P$ is an atomic cross-chain swaps protocol.
	If every party follows $P$, the swap can be completed in time $2(diam(D)+1)\Delta + 2\varepsilon$.
	Even if any set of parties deviates from $P$, the swap finishes in at most time $(diam(D)+n+1)\Delta + 2\varepsilon$. 
	Protocol $P$ requires space complexity (the total number of bits on all the blockchains) of $O(|A|\cdots|V|)$.
	Protocol $P$ requires local time complexity (the maximum execution time of a contract in a swap to transfer the corresponding asset) of $O(|V|)$.
\end{thm}

\begin{proof}
	Protocol $P$ is atomic by Lemmas \ref{lem7} and \ref{lem8}.
	The execution time of $P$ follows from Lemmas \ref{lem5} and \ref{lem6}.
	Space complexity of $P$ is $O(|A|\cdot|V|)$ because each arc has one contract and each contract requires public keys of $|V|$ parties and hashlocks of $|L| \le |V|$ leaders.
	Local time complexity of $P$ is $O(|V|)$ because we only need to verify $|L|(\le |V|)$ secrets and at most $|V|$ signatures.
\end{proof}

\section{Discussion}\label{discussion}
We made two assumptions, in Section 2.3, on the value function specifying the values of assets to parties.
Remind that, we assume that for any subgraph $D'$ of $D$, no coalition $C$ in $D'$ gets more benefits in $D'$ than $D$ when every other party in $D'$ gets benefits in $D'$ no less than in $D$.\par
The first alternative is to replace coalition $C$ in the above assumption with just a party $v$.
That is, we assume that for any subgraph $D'$ of $D$, no party $v$ in $D'$ gets more benefits in $D'$ than $D$ when every other party in $D'$ gets benefits in $D'$ no less than in $D$.
Under this assumption, the proposed protocol $P$ is not strong Nash equilibrium but is Nash equilibrium, which can be proved in the same way as Lemma 4.8.\par

The second alternative is the assumption that Herlihy made for his protocol \cite{1}.
Specifically, He classifies the parties $v \in V$ into four groups as follows,
based on the status of the arcs entering and leaving $v$ in the end of an execution of a swap algorithm.
\begin{description}
\item[DEAL]: All arcs entering and leaving $v$ are triggered.
\item[NO\_DEAL]: No arc entering or leaving $v$ is triggered.
\item[FRRE\_RIDE]: Some arc entering $v$ is triggered, but no arc leaving $v$ is triggered.
\item[DISCOUNT]: All arcs entering $v$ are triggered, but some arc leaving $v$ is not triggered.
\item[UNDER\_WATER]: Some arcs entering $v$ is triggered and some arc leaving $v$ is not triggered.
\end{description}
He assumes that no party accepts to be UNDER\_WATER, more generally, every coalition tries to avoid the case
that any member of the coalition becomes UNDER\_WATER, with the highest priority.
He also assumes that every party prefer DEAL to NO\_DEAL
and some parties or coalitions may deviate from the swap protocol because they prefer FREE\_RIDE and DISCOUNT to DEAL. 
Our model has a sufficient power to represent this assumption.
It suffices to define the benefit of $v$ in $D'$ as  $value^+_{D',v} - value^-_{v,D'} + exp(D',v)$ where the \emph{exception value} $exp(D',v)$ is defined as follows:　$exp(D',v)= -\infty$ if there exist arc $(u,v) \in A \setminus A'$ and arc $(v,w) \in A'$, otherwise $exp(D',v)=0$.
If a party (or some party of a coalition) ends up with UNDER\_WATER, then the benefit of the party (or the coalition) is $-\infty$, thus they try to avoid the situation with the highest priority.
The situation ending up with FREE\_RIDE or DISCOUNT brings a party or a coalition a larger benefit than DEAL. 

In what follows, we show that our protocol $P$ is strong Nash equilibrium under the second assumption,
equivalent to the Herlihy's assumption\cite{1}.
Assume by contradiction that, in an execution of $P$, some coalition $C$ gets more benefits by deviating from protocol $P$ than by conforming to $P$. 
Let $D'=(V',A')$ be the subgraph of $D$ induced by the set of the arcs triggered in the execution.
Since we introduce the exception value, every party in $C$ must trigger all his entering arcs in $D$ if one of his leaving arc in $D$ is triggered.
Since $D$ is strongly connected, if there exists an arc $(u,v) \in A'$, there exists an arc $(v,w) \in A'$, thus all entering arcs to $v$ are included in $A'$.
Therefore, every entering arc of every party in $V'$ is included in $A'$, which implies $D=D'$.
This contradicts the assumption that $C$ gets more benefits in $D'$ than in $D$.

\section{Conclusions}
In this paper, we proposed an atomic cross-chain protocol to improve space complexity and local time complexity. 
Herlihy's protocol \cite{1} requires to store the swap topology in each contract and set timelocks to each secret by using the topology.
The proposed protocol need not store the swap topology in any contract.
Instead, we set the time condition to trigger a contract depending on the number of signatures sent to the contract.
Therefore, every party $v$ can immediately trigger the contracts of all entering arcs to $v$, once a contract on some leaving arc from $v$ is triggered.
This is because the signature of $v$ is not included in the signatures used to trigger the leaving contract.

%
%
%

\end{document}